\documentclass[sigconf]{acmart}
\pdfoutput=1
\renewcommand\footnotetextcopyrightpermission[1]{} 
\usepackage{booktabs} 
\usepackage{amsmath,amssymb,amsthm}    
\usepackage{float}

\newcommand{\bi}{\begin{itemize}}
\newcommand{\ei}{\end{itemize}}
\newcommand{\beq}{\begin{equation}}
\newcommand{\eeq}{\end{equation}}
\newcommand{\bqn}{\begin{eqnarray*}}
\newcommand{\eqn}{\end{eqnarray*}}
\newcommand{\ba}{\begin{array}}
\newcommand{\ea}{\end{array}}
\newcommand{\bs}{\begin{small}}
\newcommand{\es}{\end{small}}
\newcommand{\nn}{\nonumber}

\DeclareMathOperator*{\argmax}{arg\,max}

\begin{document}
\title{A Coupon-Collector Model of Machine-Aided Discovery}
\titlenote{This work was supported in part by the IBM-Illinois Center for Cognitive Computing Systems Research (C3SR), a research collaboration as part of the IBM Cognitive Horizons Network.}

\author{Aditya Vempaty}
\affiliation{%
  \institution{IBM T.~J.\ Watson Research Center}
  \streetaddress{1101 Kitchawan Road}
  \city{Yorktown Heights} 
  \state{New York} 
}
\email{avempat@us.ibm.com}

\author{Lav R. Varshney}
\affiliation{%
  \institution{Department of Electrical and Computer Engineering,}
  \streetaddress{University of Illinois at Urbana-Champaign,}
  \city{Urbana} 
  \state{Illinois} 
}
\email{varshney@illinois.edu}

\author{Pramod K. Varshney}
\affiliation{%
  \institution{Department of Electrical Engineering and Computer Science,}
  \streetaddress{Syracuse University,}
  \city{Syracuse} 
  \state{New York}}
\email{varshney@syr.edu}

\begin{abstract}
Empirical studies of scientific discovery---so-called Eurekometrics---have indicated that the output of exploration proceeds as a logistic growth curve. Although logistic functions are prevalent in explaining population growth that is resource-limited to a given carrying capacity, their derivation do not apply to discovery processes.  This paper develops a generative model for logistic \emph{knowledge discovery} using a novel extension of coupon collection, where an explorer interested in discovering all unknown elements of a set is supported by technology that can respond to queries. This discovery process is parameterized by the novelty and quality of the set of discovered elements at every time step, and randomness is demonstrated to improve performance. Simulation results provide further intuition on the discovery process.
\end{abstract}


\maketitle

\section{Introduction}
\label{sec:mot}
Eurekometric studies quantify the process of discovery and have found that the empirical number of scientific discoveries as a function of time initially increases exponentially before saturating to a \emph{capacity} limit according to an approximate logistic curve \cite{ArbesmanL2010,Arbesman2011,ArbesmanC2011}.  This is seen in the discovery of mammalian species, chemical elements, and minor planets \cite{Arbesman2011}.  The logistic curve is often used to describe population growth that is resource-limited to a given carrying capacity \cite{Sayre2008}, but population growth dynamics do not seem to describe scientific discovery.  This paper develops a generative model for discovery that recreates empirical observations, by generalizing the coupon collector model \cite{Kobza2007,Boneh1997,Read1998,BerenbrinkS2009}.  In particular, the units of discovery (species, elements, planets) are the coupons and an explorer is trying to collect (discover) them all.

There are two main generalizations introduced into the standard coupon collector formulation.

Scientific discovery problems can be cast as big-D discovery,\footnote{This terminology is inspired by the terms \emph{little-c creativity} which means novelty with respect only to the mind of the individual concerned and \emph{big-C Creativity} which means novelty with respect to the whole of previous history or eminent creativity \cite{Boden2004}.} where there is also growing interest in developing artificial intelligence (AI) support \cite{Valdes-Perez1999, Langley2000,King2009}.  One can, however, just as easily think about little-d knowledge discovery problems such as looking for restaurants on Yelp,\footnote{http://www.yelp.com/} where one might `discover' a new restaurant.  In addition to being novel, the restaurant should also be of high quality and have good ratings.  In the sequel, we will therefore evaluate the discovery problem using notions of both novelty and quality, e.g.\ as judged by a suitably knowledgeable social group \cite{Sawyer2012}. 

In the economics, cognitive psychology, and AI literatures, problem solving is described in terms of searching a problem space, where the first step is discovery to determine available elements in the space \cite{Anderson2010}. In rationality-based economic theory, the set of possible alternatives is known \cite{MarchS1958}, but under the bounded rationality framework \cite{MarchS1958,Simon1982}, the set of alternatives is not known \emph{a priori}. Thus limited agents must perform discovery, perhaps with technological support to overcome limitations.\footnote{In computational creativity, a technology supporting people in determining the set of alternatives is called a \emph{coach} \cite{Lubart2005}.}  In scientific discovery, new technologies for separating compounds into constituent parts enable discovery of new chemical elements and new telescopes enable discovery of new minor planets, for example.  In the sequel, we will consider discovery as supported by a technology, which may itself be noisy. 

\section{Problem Formulation}
\label{sec:prob}
The discovery problem is modeled as follows. Let the universe of $M$ total elements be denoted $\Theta$. The explorer is initially only aware of a subset $\Theta_0\subset\Theta$ of these $M$ possibilities. She learns of other elements in steps, where at each step she learns of a new object. In taking the step, a previously unknown object is drawn from the $M$ possible values according to the prior probability mass function (pmf) $\mathbf{p}=[p_1,\ldots,p_M]$, but the supporting technology yields noisy observations of which kind of object it is. The explorer updates her element set accordingly. Note that the technology presents its outcome to the explorer without any information regarding the initial known subset of elements $\Theta_0$. In particular, let the underlying element $\theta_t$ be investigated at time step $t$. The technology makes an observation $x_t^M$ corrupted by noise, which is i.i.d.\ across time, and infers the element to be $\hat{\theta}_t\in\Theta$ as follows: $$\hat{\theta}_t=\argmax_{\theta\in\Theta}p(\theta|x_t^M),$$
where $p(\theta|x_t^M)$ is the posterior distribution given the observation. The explorer's updated element set is $\Theta_t=\Theta_{t-1}\cup\hat{\theta}_t$. Note that a new element gets added only if it is not already known by the explorer, who can access the true value. Fig.~\ref{fig:sys} summarizes this technology-aided discovery process.

Each new element has two basic characteristics: novelty and quality. These are detailed in turn in the next sections, in terms of the size and total quality of the discovered set.

\begin{figure}
  \centering
    \includegraphics[height=1.1in]{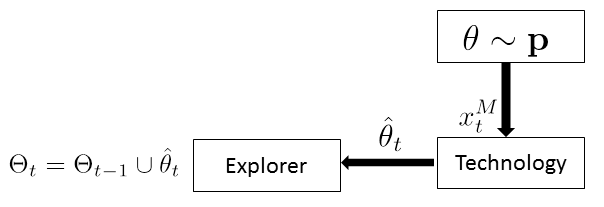}
    \caption{Single iteration of the discovery process, where $\Theta_{t-1}\subset\Theta$ is known at $t$, technology provides a noisy estimate $\hat{\theta}$ of unknown element $\theta$, and $\hat{\theta}$ is discovered by the explorer if not already known.}
    \label{fig:sys}
\end{figure}

\section{Size of Knowledge Base}
\label{sec:novel}
Size is measured by the total number of distinct elements that are known to the explorer at any given time. This includes both the prior set of elements known before the start of the discovery process and the elements discovered during discovery. Let $N_t=|\Theta_t|$. Clearly the learning rate depends on the pmf $\mathbf{p}$, the initial set $\Theta_0$, and the noise distribution. While performance is best for uniform prior when the initial knowledge set $\Theta_0$ is arbitrary, for cases when $\Theta_0$ contains rare elements, knowledge discovery is fast since the unknown elements are of high probability and can be quickly found by the explorer. This intuition is also observed in our analytical results. Also, for noisy observations, it is possible that higher variance yields faster total discovery as noise can cause false decisions in the technology aiding the discovery of new elements (typically of low probability) by the explorer.

Let the misclassification probabilities due to noise-corrupted observations at any time $t$ be $r_{mn}\triangleq \Pr(\hat{\theta}_t=n|\theta_t=m)$ for $m,n=1,\ldots,M$. Therefore, the posterior distribution $\mathbf{\tilde{p}}=[\tilde{p}_1,\ldots,\tilde{p}_M]$ of the elements is the new prior distribution of the explorer's input. An asymmetric channel $\mathbf{R} \triangleq \{r_{mn}\}$ represents the difficulty of identifying elements. As has been noted in empirical studies, some elements may be more difficult to identify/discover than others \cite{ArbesmanL2010,Arbesman2011,ArbesmanC2011}.

\begin{proposition}
\label{prop:gen}
For the technology-aided discovery problem with a pmf $\mathbf{p}$ and noise channel $\mathbf{R}$, the expected number of elements learned after $T$ time steps is 
\begin{equation}
E[N_T|\Theta_0]=|\Theta_0|-\sum_{k=1}^T\sum_{\theta\notin\Theta_0}(-1)^k{T \choose k}\tilde{p}_\theta^k.
\label{eq:final}
\end{equation}
\end{proposition}
\begin{proof}
The number of elements are updated as follows
\begin{equation}
N_t=
\begin{cases}
N_{t-1}&\text{if $\hat{\theta}_t\in\Theta_{t-1}$}\\
N_{t-1}+1&\text{if $\hat{\theta}_t\notin\Theta_{t-1}$},
\end{cases}
\end{equation}
resulting in a Markov Chain:
\begin{equation}
N_t=
\begin{cases}
N_{t-1}&\text{with probability $\sum_{\theta\in\Theta_{t-1}}\tilde{p}_{\theta}$}\\
N_{t-1}+1&\text{with probability $1-\sum_{\theta\in\Theta_{t-1}}\tilde{p}_{\theta}$}.
\end{cases}
\end{equation}
Therefore, 
\begin{eqnarray}
E[N_t|\Theta_{t-1}]&=&|\Theta_{t-1}|\sum_{\theta\in\Theta_{t-1}}\tilde{p}_{\theta}+(|\Theta_{t-1}|+1)(1-\tilde{p}_{\theta})\nn\\
&=&|\Theta_{t-1}|+1-\sum_{\theta\in\Theta_{t-1}}\tilde{p}_{\theta}.
\end{eqnarray}
For $t=1$, 
\begin{equation}
E[N_1|\Theta_0]=m_0+1-\sum_{\theta\in\Theta_0}\tilde{p}_{\theta},
\end{equation}
and for $t=2$, 
\begin{equation}
E[N_2|\Theta_1]=N_1+1-\sum_{\theta\in\Theta_1}\tilde{p}_{\theta}=N_1+1-\sum_{\theta\in\Theta_0}\tilde{p}_{\theta}-\sum_{\theta\in\Theta_1\backslash\Theta_0}\tilde{p}_{\theta}.
\label{eq:t=2}
\end{equation}
Recursively, we have for $t\geq2$,
\begin{eqnarray}
&E[N_t|\Theta_{t-2}]=E\left[N_{t-1}+1-\sum_{\theta\in\Theta_{t-2}}\tilde{p}_{\theta}-\sum_{\theta\in\Theta_{t-1}\backslash\Theta_{t-2}}\tilde{p}_{\theta}\Big|\Theta_{t-2}\right]\nn\\
&=|\Theta_{t-2}|+2\sum_{\theta\notin\Theta_{t-2}}\tilde{p}_{\theta}-E\left[\sum_{\theta\in\Theta_{t-1}\backslash\Theta_{t-2}}\tilde{p}_{\theta}\Big|\Theta_{t-2}\right].\label{eq:t=0to2}
\end{eqnarray}
An additional element is added into $\Theta_t$ only when it is not originally present in $\Theta_{t-1}$, which is with probability $\sum_{\theta\notin\Theta_{t-1}}\tilde{p}_\theta$. Under this condition that an element has been added, each $\theta\notin\Theta_{t-2}$ occurs with probability $\tilde{p}_\theta/\sum_{\theta\notin\Theta_{t-2}}\tilde{p}_\theta$, which reduces \eqref{eq:t=0to2} to the following
\begin{equation}
E[N_t|\Theta_{t-2}]=|\Theta_{t-2}|+2\sum_{\theta\notin\Theta_{t-2}}\tilde{p}_\theta-\sum_{\theta\notin\Theta_{t-2}}\tilde{p}_\theta^2.
\end{equation}
Continuing further, for a general $T$
\begin{equation}
E[N_T|\Theta_0]=|\Theta_0|-\sum_{k=1}^T\sum_{\theta\notin\Theta_0}(-1)^k{T \choose k}\tilde{p}_\theta^k.
\end{equation}
\end{proof}

Some points to be noted from Proposition~\ref{prop:gen} include the following:
\begin{itemize}
\item If rare elements are already known, the convergence is faster, cf.~\eqref{eq:final}.
\item For uniformly distributed prior pmf and perfect observations, 
\begin{equation}
E[\rho_T|\rho_0]=1-(1-\rho_0)\left(1-\frac{1}{M}\right)^T\label{eq:rho},
\end{equation}
where $\rho_t=N_t/M$.
\item Rate of knowledge discovery (rate of convergence of $E[\rho_t|\rho_0]$ to 1) is monotonically increasing in $\rho_0$.
\item Growth rate of knowledge discovery is exponential in $T$: \begin{equation}
\lim_{T\to\infty}\frac{1}{T}\log(1-E[\rho_T|\rho_0])=-\log\left(1-\frac{1}{M}\right).\label{eq:rate}
\end{equation}
Therefore, the convergence is faster as $M$ decreases.
\end{itemize}

To gain some further insight, we perform simulations.  The following parameters are used: $$\Theta=\{1,\ldots,M\},$$ $$p_\theta(m)={M-1 \choose m-1}p^{m-1}(1-p)^{M-m}$$ for $m=1,\ldots,M$, and the noise channel is an $M$-ary symmetric channel with crossover probability $r$. Fig.~\ref{fig:gen_corroborate} shows the numerical and simulation results that corroborates our analytical expressions. The simulation results are averaged over $N_{mc}=500$ Monte Carlo runs. As can be observed from Fig.~\ref{fig:gen_init}, the performance depends on the initial set $\Theta_0$ and the performance (in terms of discovery rate) is higher (and also better than the uniform case) when the initial elements are the least probable ones ($\Theta_0=[3 ,4]$ in this example). This is intuitive since if the explorer already knows the rare elements, she can discover others at a faster rate.

\begin{figure}
  \centering
    \includegraphics[height=2in, width=3.25in]{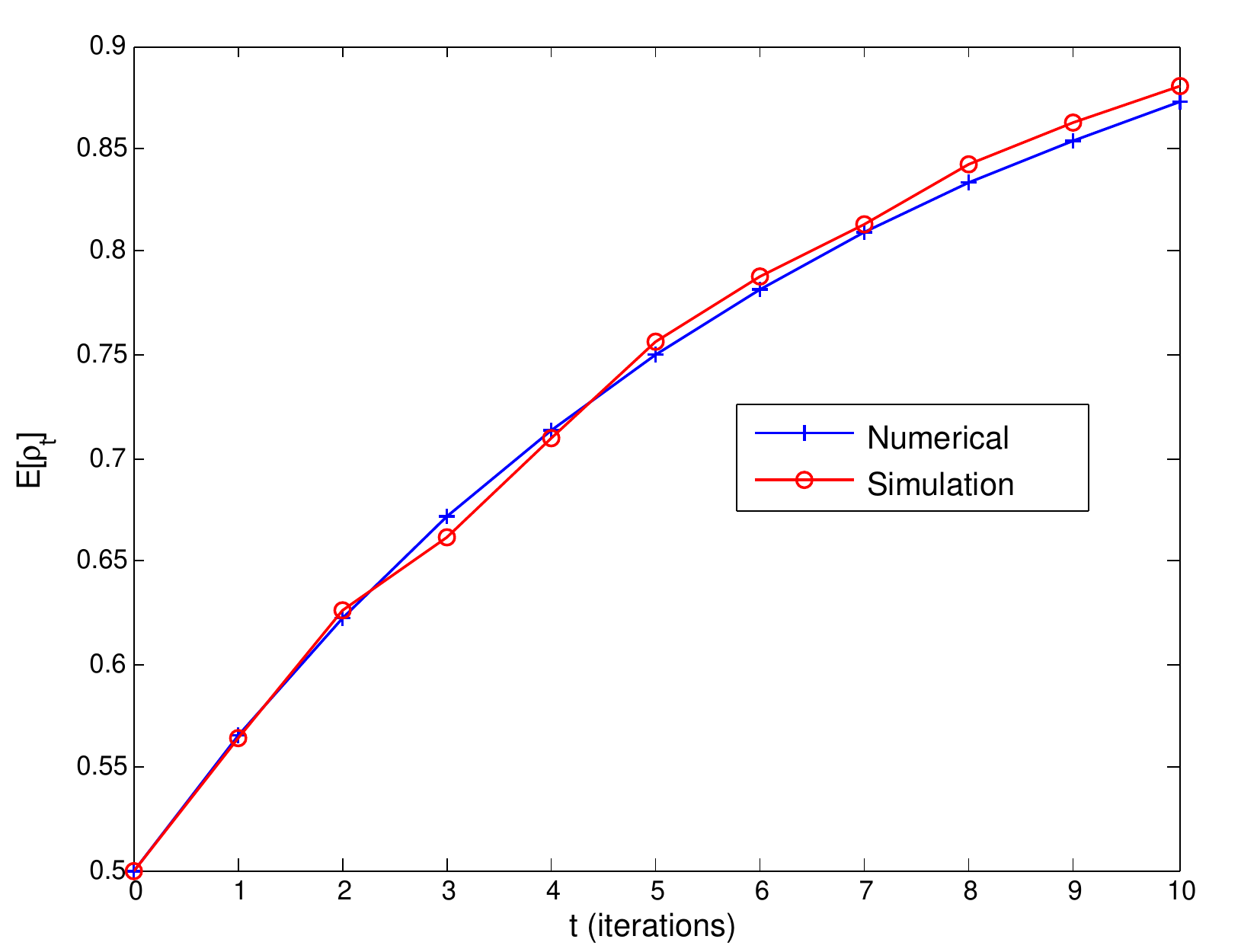}
    \vspace{-0.5cm}
    \caption{Numerical and simulations results of fraction of element set with iterations ($M=4$, $\Theta_0=[1, 2]$, $p=0.2$, and $r=0.1$).}
    \label{fig:gen_corroborate}
\end{figure}

\begin{figure}
  \centering
    \includegraphics[height=2in, width=3.25in]{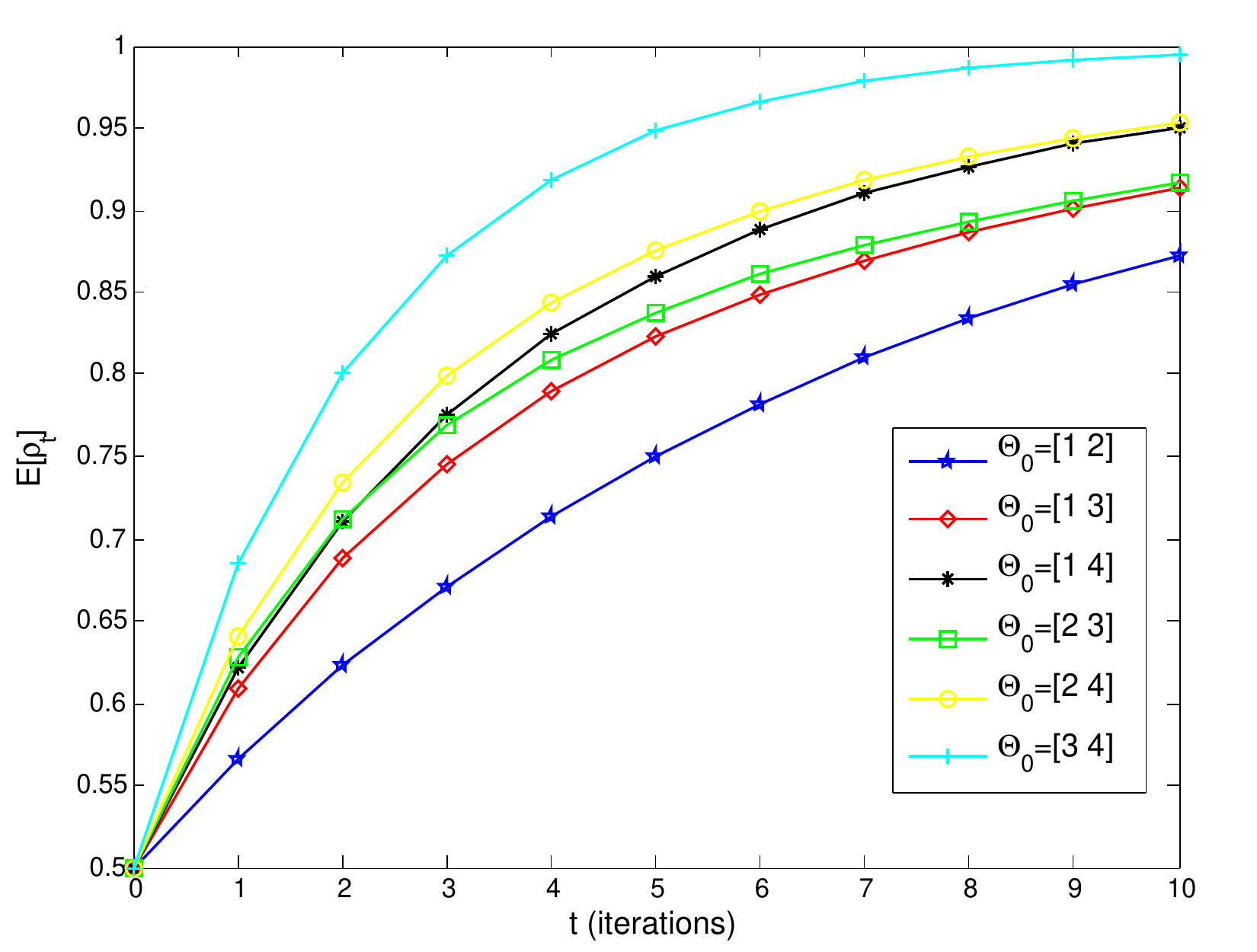}
    \vspace{-0.5cm}
    \caption{Performance variation with different initial sets $\Theta_0$ ($M=4$, $p=0.2$, and $r=0.1$).}
    \label{fig:gen_init}
\end{figure}

Fig.~\ref{fig:gen_q} shows the performance with varying noise value $r$ where higher $r$ implies noisier data produced by the technology. When the explorer is already aware of the most probable elements (refer to Fig.~\ref{fig:gen_q}), noise has the positive effect of helping the explorer in discovering the less probable elements at a faster rate, implying that \emph{noisy technology can help in discovering new elements}.  Such an observation is conceptually related to the phenomenon of noise-enhanced signal processing \cite{ChenVV2014}, where the addition of noise can improve the performance of some non-linear suboptimal inference systems. Further analysis is needed to mathematize this relation and understand the \emph{optimal} noise to yield best performance. 

\begin{figure}
  \centering
    \includegraphics[height=2in, width=3.25in]{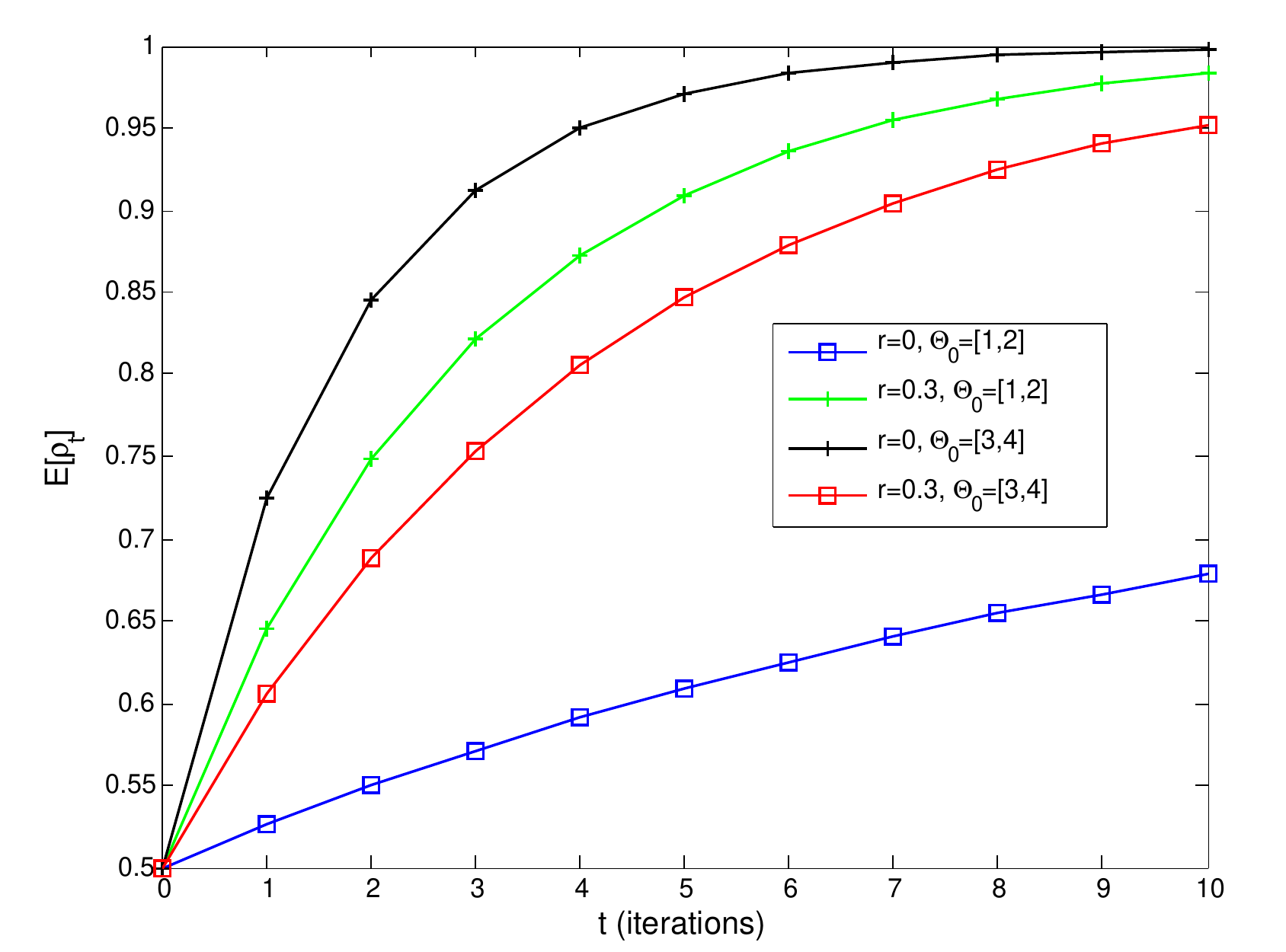}
    \vspace{-0.5cm}
    \caption{Effect of noise ($M=4$, and $p=0.2$).}
    \label{fig:gen_q}
\end{figure}

\begin{figure}
  \centering
    \includegraphics[height=2in, width=3.25in]{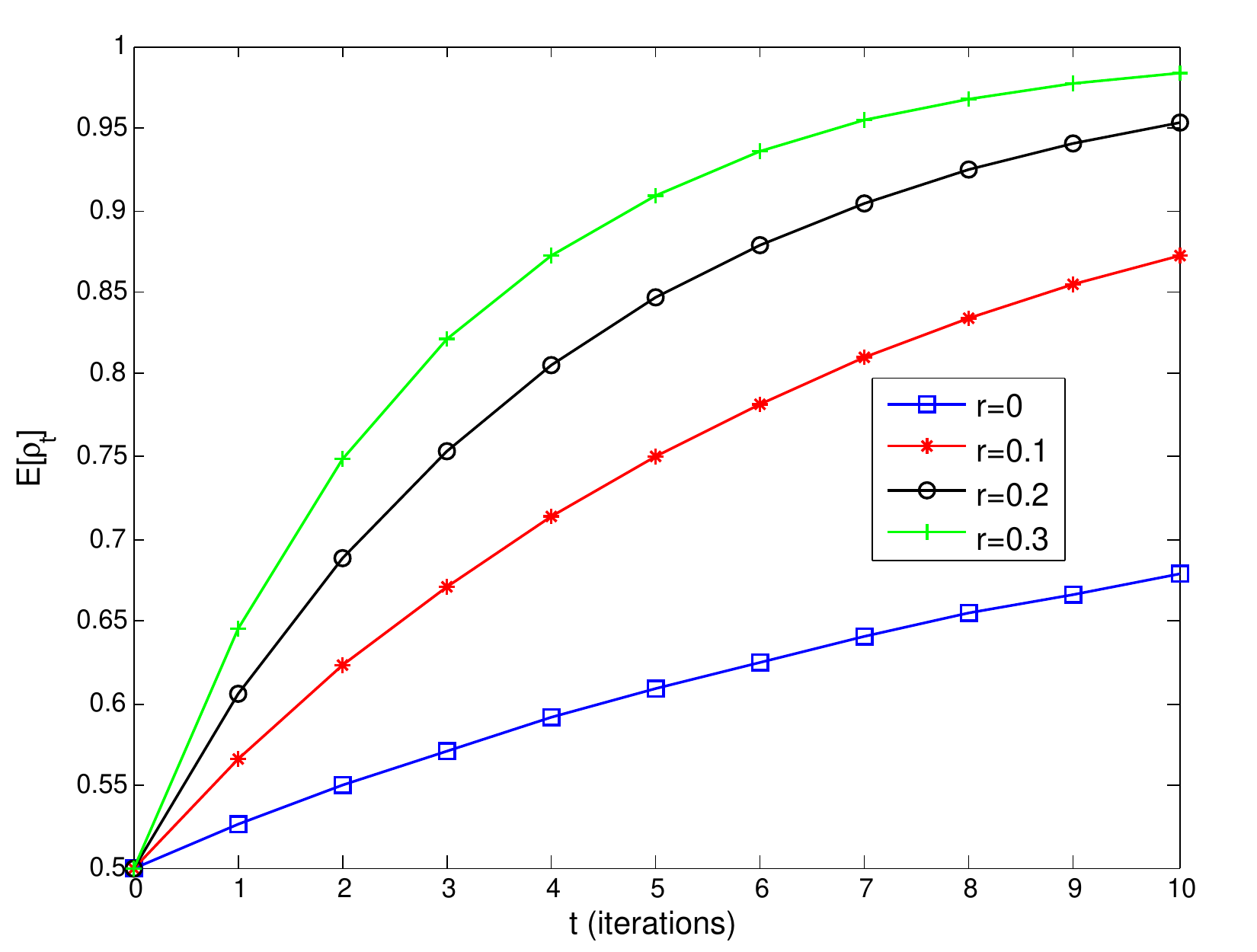}
    \vspace{-0.5cm}
    \caption{Effect of noise ($M=4$, $p=0.2$, and $\Theta_0=[1, 2]$).}
    \label{fig:gen_q1}
\end{figure}

Fig.~\ref{fig:gen_p} shows how the performance varies with prior distribution. 
The effect of prior distribution on the performance is characterized by \eqref{eq:final}, but has been omitted for brevity.%

\begin{figure}
  \centering
    \includegraphics[height=2in, width=3.25in]{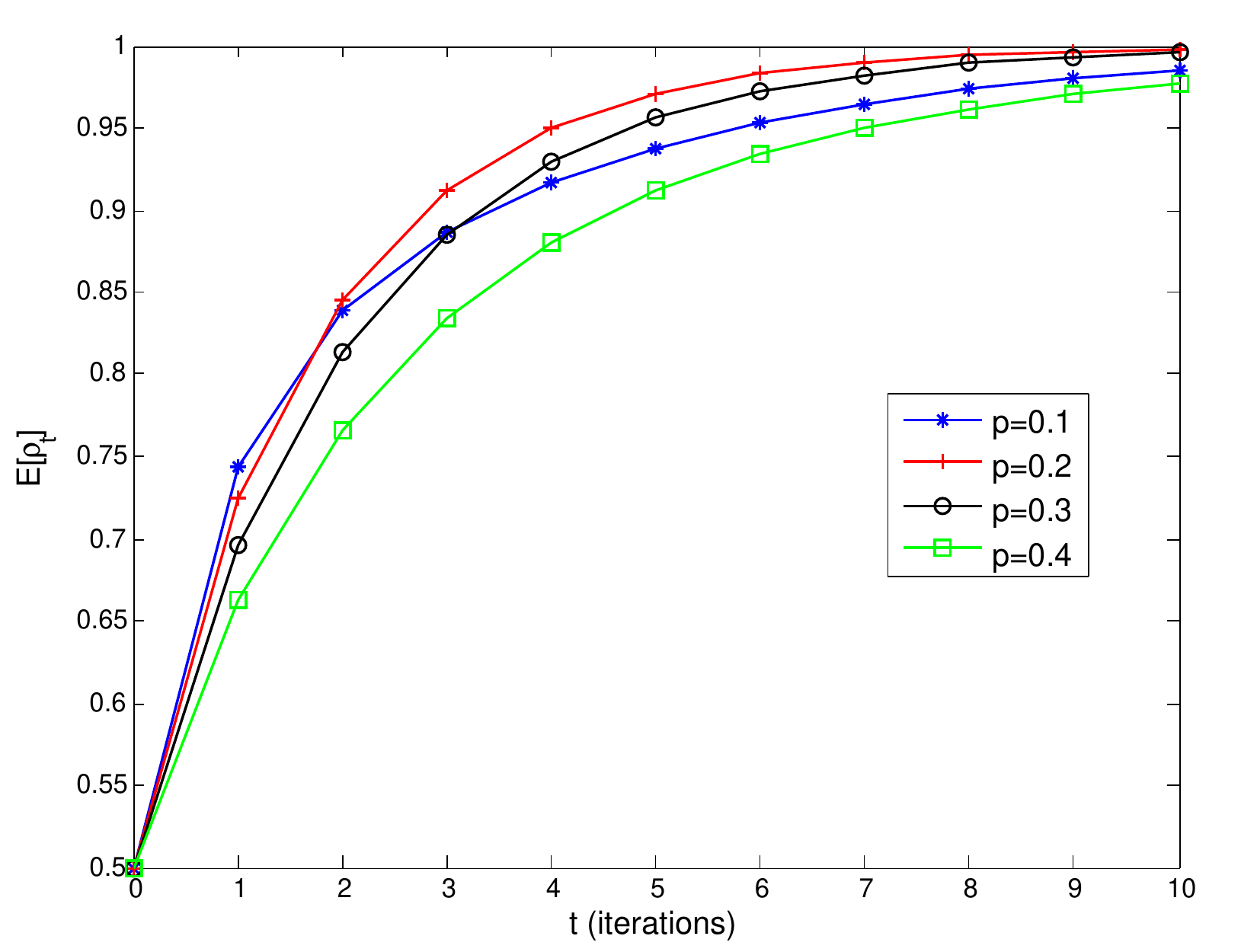}
    \vspace{-0.5cm}
    \caption{Effect of different prior distribution ($M=4$, $r=0$, and $\Theta_0=[3, 4]$).}
    \label{fig:gen_p}
\end{figure}
For the simple case, we evaluate the performance with varying $M$ and $\rho_0$ in Fig.~\ref{fig:plot_M} and Fig.~\ref{fig:plot_0}, respectively.

\begin{figure}
  \centering
    \includegraphics[height=2in, width=3.25in]{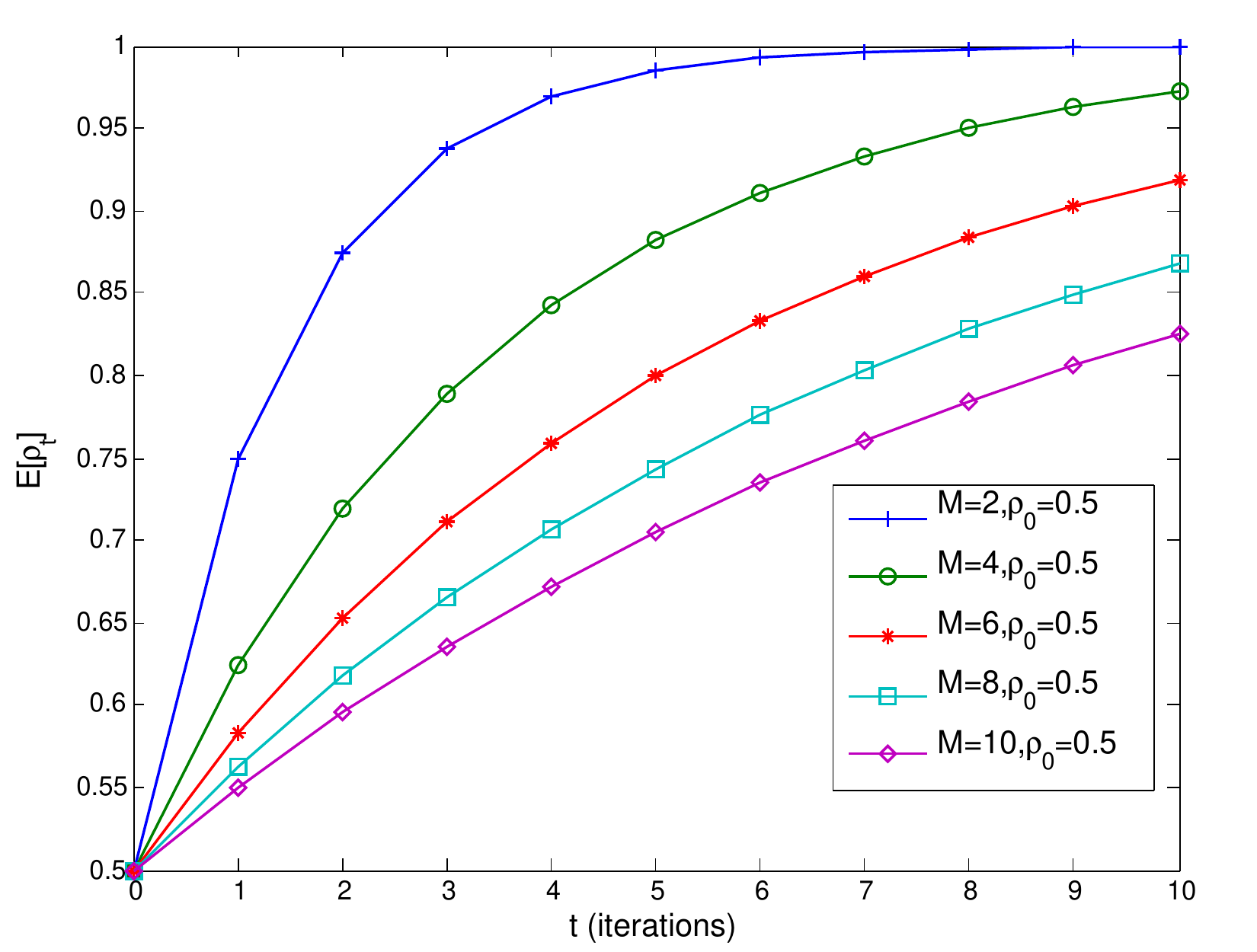}
    \vspace{-0.5cm}
    \caption{Expected fraction of element set with iterations ($\rho_0=0.5$, varying $M$)}
    \label{fig:plot_M}
\end{figure}

\begin{figure}
  \centering
    \includegraphics[height=2in, width=3.25in]{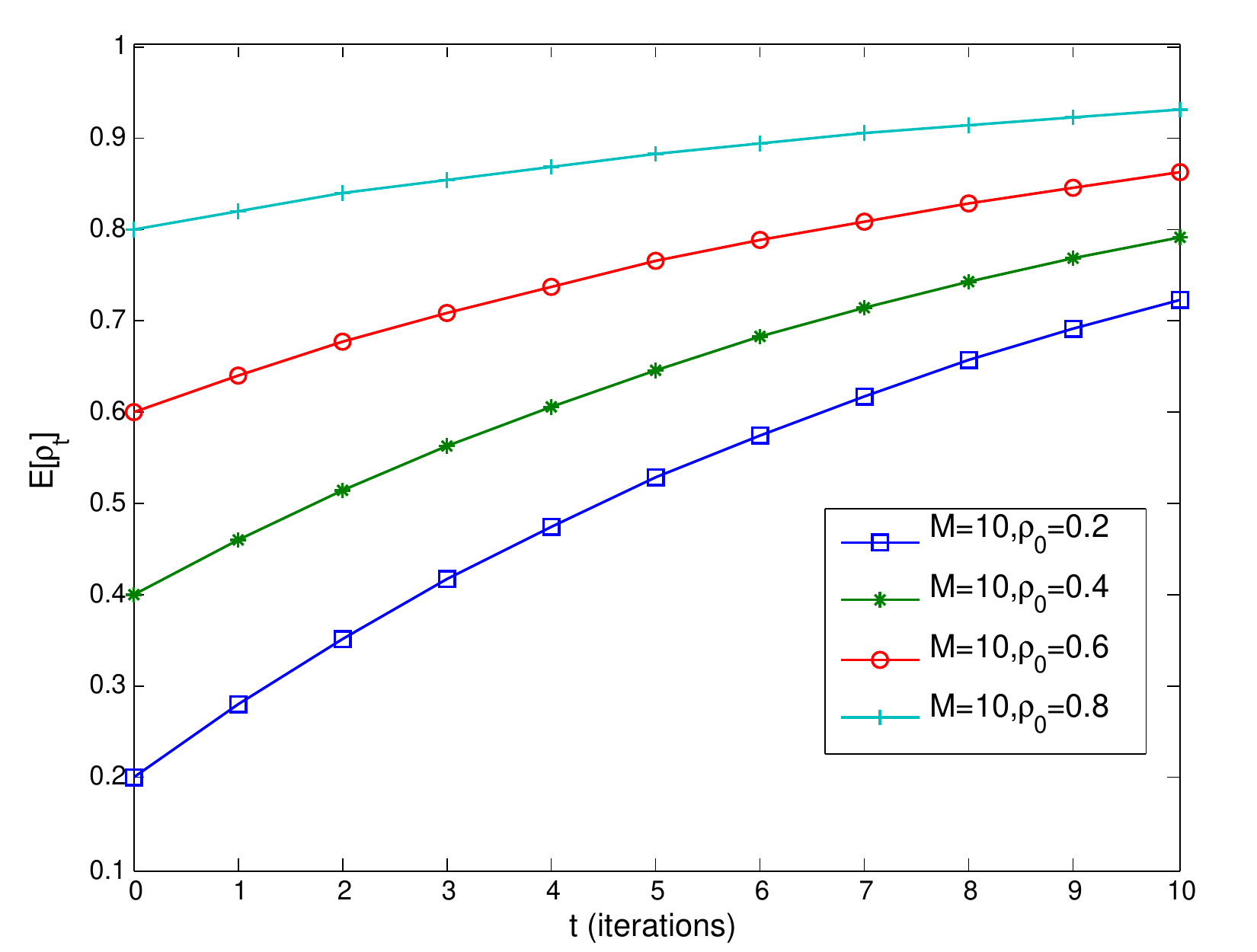}
    \vspace{-0.5cm}
    \caption{Expected fraction of element set with iterations ($M=10$, varying $\rho_0$)}
    \label{fig:plot_0}
\end{figure}

\section{Quality of Knowledge Base}
\label{sec:quality}
In many settings, not only is the discovery of a new element important but so is the quality of that element. For example, certain new metals are more useful in metallurgy than others.  Let each element $\theta \in \Theta$ have a corresponding quality factor $q_\theta$ which determines the value of discovering $\theta$. Let $\mathbf{q}=[q_1,\ldots,q_M]$ be  the quality vector and $Q_t=\sum_{\theta\in\Theta_t}q_\theta$ denote the quality of the elements discovered after time $t$. 

\begin{definition}
The $k$th-order quality-prevalence function $D^k_{\Theta_0}(\mathbf{\tilde{p}},\mathbf{q})$ is defined as the inner product between $\mathbf{q}$ (quality vector) and element-wise $k$th power of $\mathbf{\tilde{p}}$ (probability vector) over the subset $\Theta_0^C$ ($\Theta\backslash\Theta_0$): $$D^k_{\Theta_0}(\mathbf{\tilde{p}},\mathbf{q}):=\sum_{\theta\notin\Theta_0}q_\theta \tilde{p}^k_\theta.$$
\end{definition}
These functions evaluate the degree of alignment between the probability vector and the quality vector.

\begin{proposition}
\label{prop:quality}
For technology-aided discovery, the expected quality of elements discovered after $T$ iterations is given by
\begin{equation}
E[Q_T|\Theta_0]=Q_0-\sum_{k=1}^T(-1)^k{T \choose k}D^k_{\Theta_0}(\mathbf{\tilde{p}},\mathbf{q}).\label{eq:quality}
\end{equation}
\end{proposition}
\begin{proof}
The proof is similar to that of Prop.~\ref{prop:gen} with the following weighted form
\begin{eqnarray}
E\left[Q_t|\Theta_{t-1}\right]&=&Q_{t-1}\sum_{\theta\in\Theta_{t-1}}\tilde{p}_{\theta}+\sum_{\theta\notin\Theta_{t-1}}(Q_{t-1}+q_\theta)\tilde{p}_{\theta}\nn\\
&=&Q_{t-1}+\sum_{\forall\theta}q_{\theta}\tilde{p}_\theta-\sum_{\theta\in\Theta_{t-1}}q_\theta\tilde{p}_{\theta}.
\end{eqnarray}

Therefore, for a general $T$, we get the desired result
\begin{equation}
E\left[Q_T|\Theta_0\right]=Q_0-\sum_{k=1}^T\sum_{\theta\notin\Theta_0}(-1)^k{T \choose k}q_\theta\tilde{p}_\theta^k.
\end{equation}
\end{proof}

The two extreme possibilities are when $\mathbf{q}$ is aligned in the same direction as probability vector $\mathbf{p}$ or is in the opposite direction. For the previous example, Fig.~\ref{fig:quality} confirms our understanding in showing the quality of the discovered elements for these two extreme cases.

\begin{figure}
  \centering
    \includegraphics[height=2in, width=3.25in]{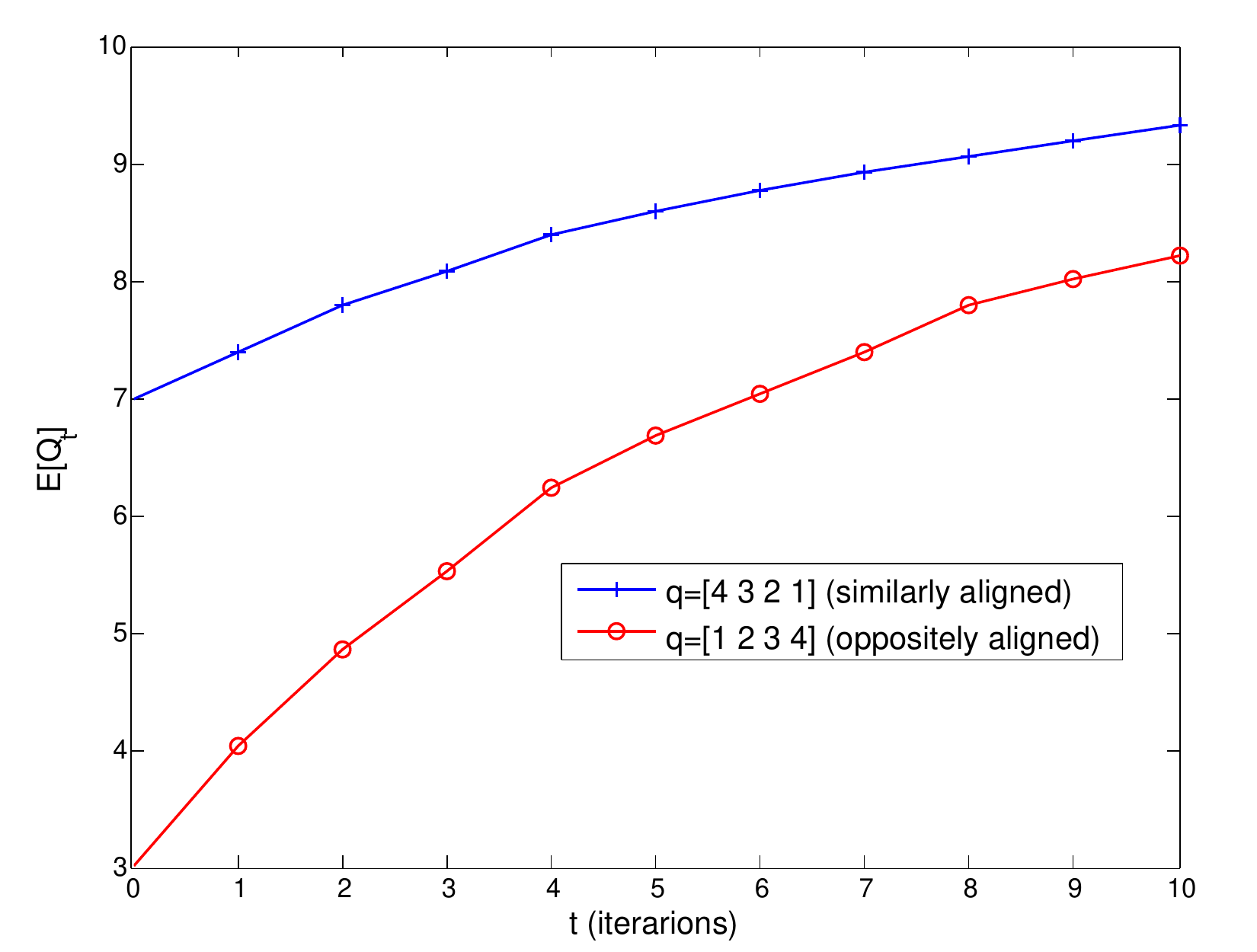}
    \vspace{-0.5cm}
    \caption{Expected quality with time ($M=4$, $r=0.1$, $p=0.2$, and $\Theta_0=[1, 2]$).}
    \label{fig:quality}
    \end{figure}

\section{Connecting to Empirical Observations}
\label{sec:emp}
As mentioned before, Eurekometrics studies the nature of discovery \cite{ArbesmanL2010,Arbesman2011,ArbesmanC2011}. Empirical observations in this field have shown the discovery of scientific output to follow an exponential increase, or more properly, a logistic growth curve. In our generative model, the ease of discovery corresponds to the pmf $\mathbf{p}$ that denotes the difficulty associated with discovering an element. When all elements are of equal difficulty, empirical results suggest a logistic curve for the number of discovered elements $D_T$ \cite{Arbesman2011}:
\begin{eqnarray}
D_T&\approx&\frac{K}{1+Ae^{-r_0T}}
\end{eqnarray}
where $K$ is the limiting size, $A$ is the fitting constant, and $r_0$ is the growth rate of scientific output.
For a small value of $A$, this can be approximated as
\begin{equation}
D_T\approx K(1-Ae^{-r_0T})\iff\rho_T^{emp}\approx1-Ae^{-r_0T}\label{eq:rho_emp}
\end{equation}
where $\rho_T^{emp}=D_T/K$ is the fraction of discovered elements. Observe that \eqref{eq:rho_emp} matches \eqref{eq:rho} derived for the generative model of the discovery process. 

\section{Discussion}
\label{sec:disc}
The discovery problem considered herein is a novel generalization of the coupon collector's problem \cite{Kobza2007,Boneh1997,Read1998,BerenbrinkS2009}. Most prior results for weighted coupon collector's problem (or coupon collector problem in general) consider the expected number of iterations required to collect all coupons, whereas here we evaluate the average number of coupons collected after $T$ iterations. Also, most of the existing results of coupon collectors problem are approximations/asymptotic order results. Moreover, to the best of our knowledge, there are no previous results on the coupon collector's problem in the presence of noisy observations or where each coupon is of different quality.

The generative model for the discovery problem has also been simulated and several intuitive observations have been made. The coupon collector's model helps interpret several empirical observations made in Eurekometrics \cite{ArbesmanC2011} and provides further insights into the discovery process, e.g.\ how discoveries become more difficult over time.

\bibliographystyle{ACM-Reference-Format}
\bibliography{HMI_lib,abrv,conf_abrv} 

\end{document}